\documentclass[12pt]{amsart}
\usepackage{amssymb}
\usepackage{amsthm}

\newcommand\mut[1]{\ignorespaces}

\newtheorem{teo}{Theorem}

\newtheorem{lema}{Lemma}

\newtheorem{defi}{Definition}
\newtheorem{coro}{Corollary}

\renewcommand\paragraph[1]{\subsection*{#1}}

\newcommand\N{\mathbb N}
\newcommand\F{\mathbb F}
\title{BCMCC \\ Exercises 3}

\usepackage{helvet}
\usepackage{courier}
\usepackage[usenames]{color}
\usepackage{xcolor}

\newcommand{\cu}{\ensuremath{\mathbf{u}}}
\newcommand{\cv}{\ensuremath{\mathbf{v}}}

\newcommand{\vs}{\ensuremath{{\rm{I\!F}}_q^n}}

\begin{document}

\title[Separating properties in Reed-Somonon codes]{A Study of the Separating Property in Reed-Solomon Codes by Bounding the Minimum Distance}

\author{Marcel Fernandez \and Jorge J. Urroz}

\email{marcel@entel.upc.edu, jorge.urroz@upc.edu}

\maketitle

\begin{abstract}
According to their strength, the tracing properties of a code can be categorized as frameproof, separating, IPP and TA.	It is known that if the minimum distance of the code is larger than a certain threshold then the TA property implies the rest.  Silverberg et al. ask if there is some kind of  tracing capability left  when the minimum distance  
falls below the  threshold. 
Under different assumptions, several papers have given a negative answer to the question.
In this paper further progress is made. We establish values of the minimum distance for which  Reed-Solomon codes do not posses the separating property.

\end{abstract}

\section{Introduction.}

As a motivation for our work, consider the distribution of digital goods. In the trade of digital content, safe guarding ownership rights is certainly a critical issue. A way to protect copyright consists of making each  copy of the content  unique. This is done by 
embedding a different mark in each delivered item. These hidden marks are typically  strings of symbols. However,
since now all objects are different, traitor users can get together and by comparing their copies, they 
create a new copy that tries to disguise their identities. This is known as a collusion attack and the newly  created copy  is usually called a 
pirate copy.

A way to deal with collusion attacks is by taking the embedded symbol strings to be the code words
of a code with tracing properties. There is a large literature about codes possessing different degrees of robustness against collusion attacks. Let us give a brief overview. Formal definitions will be done in subsequent sections. In a $c$-frameproof code~\cite{bs:csffddit}, a coalition of at most $c$ users can not create a pirate copy that contains the code word of  another user not in the coalition. In $c$-secure frameproof codes two disjoint coalitions of at most $c$ users can not create the
same pirate copy. It has been shown~\cite{mfk:asassfcoqa}, that the secure frameproof property is the same as the separating property~\cite{sagalovich1994trans}.
Codes with the Identifiable Parent Property (IPP) were introduced in~\cite{holl:ipp}. Informally, a code has the $c$-IPP property if
all coalitions of at most $c$ traitors that can generate the same pirate copy
have a
non-empty intersection, i.e. have a common user.
The IPP  has received considerable attention in the recent
years, having been studied by several
authors~\cite{bcekz:ippmp,bk:acoippcwei,tvtsm:ncfippc,cfjlm:cwtippfmf,gckm:perfpis}.
An even stronger property is the Traceability  property ($c$-TA). In this
case, it is guaranteed that the ``closest''  authorized copy to a given pirate copy belongs to
one of the traitors. Sufficient conditions for a  code to be a $c$-TA code are stated in \cite{staddon2001combinatorial}. 

The  work in~\cite{ssw:taipp} discusses efficient algorithms
for the  identification of traitors in schemes that use $c$-TA codes.
Let $M$ denote the size of the code.  For TA
codes,  tracing is an $O(M)$ process,
  whereas   for IPP codes tracing is more
expensive since it is an $O(\binom{M}{c})$ process.
Being the TA property stronger than the IPP, but being tracing  more
costly for the IPP, it seems reasonable to expect that by relaxing the TA
requirements one
is left with a code that, even though is no longer $c$-TA, still
possesses IPP. 
In this regard,  Silverberg et al.  asked the
following question:\\

\textit{Question 1}~\cite{ssw:taipp}: Is it the case that all $c$-IPP
Reed-Solomon codes are also $c$-TA?
\\

Although intuition might lead us to give a negative answer, in that same paper the
authors used truncated Reed-Solomon codes to  credit the exact opposite, that
is, if a Reed-Solomon code does not have the TA property then it neither has
the IPP.
Later, the work in~\cite{mfs:otrbttporsc} not only reinforced this conjecture, but proved
a stronger fact,  a Reed-Solomon code that is not $c$-TA  it is neither
$c$-secure-frameproof. Therefore, they generalized the above question to the following one:
\\

\textit{Question 2}~\cite{mfs:otrbttporsc}: Is it the case that  all $c$-SEP Reed-Solomon codes are also $c$-TA?
\\

In this paper, we supplement more evidence to this last question. The results we present will hopefully 
contribute to a  complete
understanding of the tracing properties in Reed-Solomon codes.

\section{Definitions and previous results.}
\label{sec:dapr}
Let $q$ be a prime power $p^s$ and let $\F_q$ denote the finite field with $q$ elements. 
Denoting by  $\F_q^n$ the set of all $n$-tuples with elements from $\F_q$. We define a linear code of length $n$ to be a vector subspace of $\F_q^n$. 
Then, $\F_q$ s is called the
{\it code alphabet}, and the $n$-vectors in the code are called 
{\it code words}.
The dimension of the code
is defined as the dimension of the vector subspace.
Let $\cu ,\cv \in \vs$ be two words, then the {\it Hamming distance} ${\it
d}(\cu , \cv)$ between $\cu$ and $\cv$
is the number of
positions where $\cu$ and $\cv$ differ.
The {\it minimum distance} $d$, is defined
as the smallest distance between
two different code words.
A linear code with length $n$, dimension $k$ and minimum distance $d$ is
denoted as a $[n,k,d]$-code.

Reed-Solomon codes can be defined as follows.
Let $\F_q[x]$ be the ring of polynomials over $\F_q$. Take all polynomials 
of degree less than $k$, $\F_q[x]_{k-1}\subset \F_q[x]$.
Let $\alpha$ be a primitive element of  $\F_q$, so we have 
$\{1, \alpha, \alpha^2, \ldots,\alpha^{q-2}\} =\F_q\setminus\{0\}$.
\begin{defi}
A Reed-Solomon, $RS[n,k]_q$, code is defined as the
vector subspace of  $\F_q^n$ determined by all vectors of the form
$$
\mathbf{v}=(f(0), f(1), f(\alpha),\ldots,f(\alpha^{q-2}))
$$
where $f\in \F_q[x]_{k-1}$. Note that $n = q$.
\end{defi}

As in the previous definition, throughout the paper, and probably with
an slight abuse of notation we will denote polynomials with an italic
lowercase letter.

Reed-Solomon codes are maximum distance separable (MDS)~\cite{mws:ttoecc}. That means they attain the Singleton bound with equality $d = n - k + 1$.

\subsection{Definitions about codes with tracing properties.}
\label{sec:baprocitc}
Let $C$ be an $[n,k,d]$ code over $\F_q$, let $T=\{\mathbf{t}^1,\ldots, \mathbf{t}^c\}\subseteq C$ with
$\mathbf{t}^i=(t_1^i,\ldots,t_n^i)$ be
a subset of size $c$. Also, let $T|_i = \{t_i^j | \mathbf{t}^j \in T \}$. The
\emph{descendant set}
of
$T$, is defined
as
$$
desc(T)=\left \{ \mathbf{z}=(z_1,\ldots,z_n) \in \F_q^n | z_i\in  T|_i\}, 1 \leq i \leq n \right \}.
$$

\begin{defi}
	\label{def:sep}
	A code $C$, defined  over $\F_q$, has the $c$-separating
	property (denoted $c$-SEP), $c>0$, if for any two disjoint subsets of 
	$C$, $U=(\mathbf{u}^1,\ldots, \mathbf{u}^c)$ and 
	$V=(\mathbf{v}^1,\ldots, \mathbf{v}^c)$,  we have
	$$
	U|_i  \cap V|_i = \emptyset \ \ \mathrm{for} \ \ 1 \leq i \leq n.
	$$
	
\end{defi}

In the introduction we used the name secure frame-proof for the separating property.

\begin{defi}
	\label{def:IPP}
	A code $C$, defined  over $\F_q$, has the $c$-Identifiable Parent
	Property (denoted $c$-IPP), $c>0$, if for all $\mathbf{z}\in \F_q^n$
	and for all coalitions  $T \subseteq C$ of at most $c$ code words, we have

	$$
	\mathbf{z}\not \in \bigcup_{T, |T|\leq c} desc(T) \ \ \mathrm{or} \ \
	\bigcap_{\mathbf{z}\in desc(T)} T \neq \emptyset.
	$$
\end{defi}

\begin{defi}
	\label{def:TA}
A code $C$ is a $c$-traceability code (denoted $c$-TA), for $c>0$,
if for all subsets (coalitions)  $T \subseteq C$ of at most $c$ code
words, if $\mathbf{z}\in desc(T)$, then there exists a
$\mathbf{t}\in T$ such that
$d(\mathbf{z},\mathbf{t})<d(\mathbf{z},\mathbf{w})$ for all
$\mathbf{w}\in C-T$.
\end{defi}

We will also have ocasion to link our discussion to a weaker tracing property called 
$c$-frameproof (FP).

\begin{defi}
	\label{def:fp}
	A code $C$, defined  over $\F_q$, has the $c$-Frameproof
	Property (denoted $c$-FP), $c>0$, if for any code word $\mathbf{u}$ and a subset of  
	$C$ of size at most $c$
	$V=(\mathbf{v}^1,\ldots, \mathbf{v}^c)$, with $\mathbf{u} \notin V$,  we have
	$$
	\mathbf{u}|_i  \notin V|_i \ \ \mathrm{for} \ \ 1 \leq i \leq n.
	$$
	
\end{defi}

Note that $(c,1)$-SEP is equivalent to $c$-FP.

\subsection{Bezout identity}

Some of the results in this paper, make extensive use of the Bezout identity. Intuitively, the 
Bezout identity is the ability to do the euclidean algorithm backwards.

\begin{defi}
	\label{def:bezoutiden}
	Let $u$ and $v$ two elements in a ring $A$, and let $d=(a,b)$ a greatest common divisor of $a$ and $b$. The Bezout identity is an identity of the form 
	$$
	au-by=d,
	$$
	for some elements $a$ and $b$ in the ring $A$.
\end{defi}

\subsection{The separating condition for Reed-Solomon codes.}
\label{sec:rsippcodes}

Let us state previous results that lead to the motivation of our work.
In \cite[Lemma 1.6]{staddon2001combinatorial} authors show that if $|C| > c \geq q$ then $C$ is not a
$c$-IPP code.
In \cite[Lemma 1.3]{staddon2001combinatorial} it is shown that the $c$-TA property implies the
$c$-IPP property.
In  \cite{cfn:tt,cfnp:tt2000}\cite[Theorem 4.4]{staddon2001combinatorial} it is
proved that any $[n,k,d]$ code over $\F_q$ with  $d>n-n/c^2$ and $q>c$ is a $c$-TA code.

Interestingly enough in  \cite[Theorem 8]{ssw:taipp} authors construct a family of
truncated ($n<q-1$) $RS[n,k]_q$  codes that fail to be $c$-IPP if
$c^2>n/(n-d)$.
Then in \cite[Question 11]{ssw:taipp} the authors ask if it is always
true that for Reed-Solomon codes the $c$-IPP fails if $c^2>n/(n-d)$.

This is a very interesting question because a positive answer would mean that 
for Reed Solomon the $c$-IPP and $c$-TA properties are essentially the same. 

This question has been addressed in~\cite{fcsd:anatippirsc,mfs:otrbttporsc}. We can summarize previous results

\begin{teo}[\cite{fcsd:anatippirsc}, Theorem 6]
	\label{thm:ippmultiplicative}
	 Let $RS[n, k]_q$ be a Reed-Solomon code over $\F_q$ such
	that $k-1$ divides $q-1$. Then, if $d \leq n - \frac{n}{c^2}$ the code is not $c$-SEP.
\end{teo}

\begin{coro}[\cite{mfs:otrbttporsc}, Corollary 2]
	\label{cor:sepcgsqrtk1}
	Let $C$ be a $[n,k,d]$ Reed-Solomon code over $\F_q$. If $c\geq \sqrt{q-1}$ and 
	$d \leq (1-1/c^2)n$, then the code is not $c$-SEP.
 
\end{coro} 

 \begin{teo}[\cite{mfs:otrbttporsc},Theorem 2]
 	\label{thm:ippadditive}
 	Let $RS[n, k,d]_q$ be a Reed-Solomon code over $\F_q$ and $c$ a divisor of $q$.
 	Then, if $d \leq n-\frac{n}{c^2}$ the code is not $c$-SEP.
 \end{teo}

It is worth noting that the proofs of these theorems are constructive in the sense that 
explicit disjoint sets $U,V$, such that 
$$
U \cap V = \emptyset \ \ \mathrm{and} \ \ desc(U) \cap desc(V) = \emptyset
$$
are found. Therefore proving  that for Reed-Solomon codes under the conditions of
the theorems $c$-SEP, $c$-IPP and $c$-TA are in fact equivalent.


\subsection{Our contribution}
In this paper,  progress in the understanding of the tracing properties in Reed-Solomon codes is made. In the flavour
of Theorem~\ref{thm:ippmultiplicative} and Theorem~\ref{thm:ippadditive} we use the structure of the finite field
$\F_q$, over which the code is defined. In our particular case, we take advantage of the divisors of $q-1$. With that, we are able to give a complete answer to Question 2 by proving, in a constructive way, that in Reed-Solomon codes 
$c$-IPP and $c$-TA properties are essentially the same when $q\equiv 1\pmod {c^2}$. More precisely, we set the minimum distance to $d = \lceil n-\frac{n}{c^2} \rceil$ which is the maximum allowed so the code is not $c$-TA and then find 
two disjoint sets of code words that are not separated.
In the rest of the paper, although the proofs  are also constructive,
the approach is somehow different.  We relax the distance condition and study whether a $[n,k,d]$ Reed-Solomon  code over $\F_q$ 
 with minimum distance $d< (n-r)$ is $c$-SEP for some $r>n/c^2$. To show the flavour of our approach, we start by proving a relaxed version of the question and state that Reed-Solomon codes  
are not $c$-SEP for $r=n/c$ . Then, we proceed to strengthen this result. For the case $c=2$, we answer the question for $r=\left[\frac{q}{3}\right]$ and for $c=3$ we do so for $r=2\left[\frac {q}{8}\right]$.
We round up the paper using an elegant result of Cilleruelo~\cite{cilleff} to give an alternate and more concise proof of known results.

\section{A connection with the frameproof property.}

We start our discussion by studying Reed-Solomon codes over $\F_q$ with minimum distance $d\leq q-\frac{q}{c}$. 

\begin{teo} 
	\label{thm:sepfp}
	For any $q$ a power of prime,  $c\ge 2$ and $d\leq q-
\dfrac q{c}$,   Reed-Solomon $[n,k,d]$ codes over $\F_q$ are not $c$-SEP.
\end{teo}

\begin{proof}
Let $\F_q=\{\alpha_1,\alpha_2\dots,\alpha_q\}$ and $c\ge 2$ an integer. To construct the Reed Solomon code we  consider  any set of $c$  distinct polynomials $\{f_1,\dots,f_c\}\subset \F_q[x]_{k-1}$ for $k-1=q-d$. Observe that in this case $c(k-1)\geq q$.

Let $q=cl+r$ where $0\le r<c$. Since the case $c|q$ is taken care of in Theorem~\ref{thm:ippadditive}, 
we consider $c\nmid q$ and then $r>0$. We take the maximum allowed minimum distance $d=\lfloor q-\frac q{c}\rfloor$.
Then,
$l=\frac qc-\frac rc<k-1$ and since $l$ is an integer $l\le k-2$ which means $l+1\le k-1$.

For $1\le i\le r$ we take
$$
 g_i=f_i+\prod_{m=1}^{l+1}(x-\alpha_{(l+1)(i-1)+m}),
$$ 
while for $ r+1\le i\le c$ take 
$$
g_i=f_i+\prod_{m=1}^{l}(x-\alpha_{(i-1)l+m+r}).
$$

Observe that $g_i\in\F_q[x]_{k-1}$ for all $1\le i\le c$. Then,  by construction, for every $1\le j\le q$ the polynomials  $f_{t}$ and $g_t$ with $t=[\frac{j}{k-1}]$ are such that $f_t(\alpha_j)=g_t(\alpha_j)$. Note that $t\le j/(k-1)\le q/(k-1)\le c$. Hence   $\{f_1(\alpha),\dots,f_c(\alpha)\}\cap \{g_1(\alpha),\dots,g_c(\alpha)\}\ne\emptyset$ for any $\alpha\in\F_q$ and so the two sets citeof code words given by $(f_i(\alpha_1),\dots,f_i(\alpha_n)$ and $(g_i(\alpha_1),\dots,g_i(\alpha_n)$ for $1\le i\le c$ are not separated.

We finally show that the constructed polynomials are all different. Observe that for any two polynomials $f,g$  of degree smaller than $q$ it is not possible that  $f(\alpha)=g(\alpha)$ for every $\alpha\in\F_q$, since  in that case $(x^q-x)|(f-g)$. Hence, taking $f_i\ne f_j$ for any $i\ne j$, then for any $1\le j,l\le c$ there exist an $\alpha\in\F_q$ such that $g_j(\alpha)\ne f_l(\alpha)$  and also for any for $i\ne j$, there exist a $\beta\in\F_q$ such that   $g_i(\beta)\ne g_j(\beta)$  since for any two polynomials in the code $f,g$ we have deg$(f-g)\le k-1=q-d< q$.

Finally, note that, since the code is MDS, taking $n=q$ we have $d=q-k+1$ and the condition $d\leq q-q/c^2$ implies $q\leq c^2(k-1)$ which is true in our case since $c^2(k-1)\ge c(k-1)\ge q$. 
\qed
\end{proof}

In Section~\ref{sec:baprocitc} we defined $c$-frameproof codes. Take an $[n,k,d]_q$ code. It is well known, see proof of Lemma III.2 in~\cite{bs:csffddit}, that if $d>n-\frac{n}{c}$ then the code is $c$-FP. In Theorem~\ref{thm:sepfp} we have also proved that a Reed-Solomon code with
minimum distance $d\leq n-\frac{n}{c}$ is not $c$-FP. Indeed, by taking $f_1 = f_2 = \cdots = f_c = \alpha$, $\alpha \in \F_q$ in the proof of the theorem, for every $\alpha \in \F_q$ there is a $g_i$ such that
$g_i(\alpha)= \alpha$.

\section{Increasing the minimum distance.}

In the previous section we saw that a Reed Solomon code with a small distance is not separated.
This is consistent with intuition since then the code has a higher dimension as a vector space then
chances are code words are not ``separated''. In this section we discuss strategies to increase the minimum
distance of the code and still keep non-separation.

\subsection{The case $c=2$.}

To show our approach we first deal with a particular case.

\begin{lema} The $[11,4,8]$ Reed-Solomon code over $F_{11}$  is not $(2,2)$-separating.
\end{lema}

\begin{proof}
We will find polynomials $f_1$,$f_2$ and $g_1$,$g_2$ such that the corresponding pairs of codewords $\{\mathbf{f}_1, \mathbf{f}_2\}$ and $\{\mathbf{g}_1, \mathbf{g}_2\}$ are not separated. 

Consider the polynomial $f_1=0$, and take $g_1=\gamma_1\prod_{i=1}^{3}(x-\alpha_i)$, for some  $\{\gamma_1,\alpha_1,\alpha_2,\alpha_3\}\in\F_{11}$. Now, let 
$$
f_2=\sum_{i=1}^{3}g_1(\alpha_{3+i})\frac{\prod_{j\in \{1,2,3\},j\ne i}(x-\alpha_{3+j})}{\prod_{j\in \{1,2,3\},j\ne i}(\alpha_{3+i}-\alpha_{3+j})}+\phi_2\prod_{i=1}^{3}(x-\alpha_{3+i}),
$$

for some $\{\phi_2,\alpha_4,\alpha_5,\alpha_6\}\in\F_{11}$.
Finally, consider 
$$
g_2=\gamma_2\prod_{i=1}^3(x-\alpha_{6+i}),
$$
for some $\{\gamma_2,\alpha_7,\alpha_8,\alpha_9\}\in\F_{11}$.

By construction $\{f_1(\alpha_i),f_2(\alpha_i)\}\cap \{g_1(\alpha_i),g_2(\alpha_i)\}\ne\emptyset$ for $i=1,\dots,9$. Now, selecting $\phi_2,\gamma_2$ such that
\begin{eqnarray*}
	\sum_{i=1}^{3}g_1(\alpha_{3+i})\frac{\prod_{j\in \{1,2,3\},j\ne i}(\alpha_{10}-\alpha_{3+j})}{\prod_{j\in \{1,2,3\},j\ne i}(\alpha_{3+i}-\alpha_{3+j})}&=&\gamma_2\prod_{i=1}^3(\alpha_{10}-\alpha_{6+i})-\phi_2\prod_{i=1}^{3}(\alpha_{10}-\alpha_{3+i})\\
	\sum_{i=1}^{3}g_1(\alpha_{3+i})\frac{\prod_{j\in \{1,2,3\},j\ne i}(\alpha_{11}-\alpha_{3+j})}{\prod_{j\in \{1,2,3\},j\ne i}(\alpha_{3+i}-\alpha_{3+j})}&=&\gamma_2\prod_{i=1}^3(\alpha_{11}-\alpha_{6+i})-\phi_2\prod_{i=1}^{3}(\alpha_{11}-\alpha_{3+i}),
\end{eqnarray*}
which is possible whenever 
$$
\prod_{i=1}^3(\alpha_{10}-\alpha_{6+i})(\alpha_{11}- \alpha_{3+i})\ne\prod_{i=1}^3(\alpha_{11}-\alpha_{6+i})(\alpha_{10}-\alpha_{3+i}),
$$
we get
\begin{eqnarray*}
	f_2(\alpha_{10})&=&g_2(\alpha_{10})\\
	f_2(\alpha_{11})&=&g_2(\alpha_{11})
\end{eqnarray*}
and hence the pairs $\{\mathbf{f}_1, \mathbf{f}_2\}$ and $\{\mathbf{g}_1, \mathbf{g}_2\}$ are not separated. 
\qed
\end{proof}

As an example take $\alpha_i=i$. Then 
$$
\prod_{i=1}^3(\alpha_{10}-\alpha_{6+i})(\alpha_{11}- \alpha_{3+i})=\prod_{i=1}^3(4-i)(8-i)=7!/4,
$$
while
$$
\prod_{i=1}^3(\alpha_{11}-\alpha_{6+i})(\alpha_{10}-\alpha_{3+i})=\prod_{i=1}^3(5-i)(7-i)=4\cdot6!.
$$
In this case
\begin{eqnarray*}
	g_1&=&x^3-6x^2+11x-6\\
	f_2&=& 5x^3+10x+9\\
	g_2&=& 10x^3+x^2+x+9
\end{eqnarray*}
and we have that 
\begin{eqnarray*}
	\mathbf{f}_1&=&(0,0,0,0,0,0,0,0,0,0,0)\\
	\mathbf{f}_2&=&(2,3,9,6,2,5,1,9,4,5,9)\\
	\\
	\mathbf{g}_1&=&(0,0,0,6,2,5,10,1,6,9,5)\\
	\mathbf{g}_2&=&(6,1,10,5,2,6,0,0,0,5,9)
\end{eqnarray*}
are not separated. It is clear that the solution is not unique,
Taking $\alpha_i=i+1$, we get
\begin{eqnarray*}
	g_1&=&x^3-9x^2+26x-24\\
	f_2&=& -\frac{13}{4}x^3+\frac{135}2x^2-\frac{1715}4x+\frac{1737}2\\
	g_2&=& -x^3+27x^2-242x+720,
\end{eqnarray*}
which give the pairs
\begin{eqnarray*}
	\mathbf{f}_1&=&(0,0,0,0,0,0,0,0,0,0,0)\\
	\mathbf{f}_2&=&(9,2,3,9,6,2,5,1,9,4,5)\\
	\\
	\mathbf{g}_1&=&(5,0,0,0,6,2,5,10,1,6,9)\\
	\mathbf{g}_2&=&(9,6,1,10,5,2,6,0,0,0,5)
\end{eqnarray*}
that are also not separated.

Note that the previous lemma answers Question 2 for $q=11$ and $c=2$ and for these particular values improves Theorem~\ref{thm:sepfp}. However this approach does not fully generalize. Fortunately, for general $q$ we can give a fully explicit answer, by decreasing a bit the distance of the code.

\begin{teo}\label{teo:2gen} A Reed Solomon code over $\F_{q}$ with distance $d=q-\left[\frac{q}{3}\right]$ is not $2$-SEP.
\end{teo}

\begin{proof} Let $q=3l+r$ where $0\le r\le 2$,  and consider for any $1\le i, j\le 2$, $A_{i,j}\subset \F_q$ disjoint sets such that $A_{2,2}=\emptyset$, and $|A_{i,j}|=l+1$ for $r$ of them and  $|A_{i,j}|=l$ for the remaining $3-r$ sets. Now let $p_{i,j}=\prod_{\alpha\in A_{i,j}}(x-\alpha)$. Observe that $\cup_{i,j}A_{i,j}=\F_q$.
	The following polynomials, of degree at most $l+1$ define a code that is  not $2$-SEP.
	\begin{eqnarray*}
		f_1&=&0,\\
		f_2&=&p_{2,1}-p_{1,1},\\
		g_1&=&-p_{1,1},\\
		g_2&=&-p_{1,2}.\\
	\end{eqnarray*}
	Indeed, for any $\alpha\in A_{1,1}$ we have $f_1=g_1$, for those $\alpha\in A_{1,2}$, $f_2=g_2$, and for those $\alpha\in A_{2,1}$, $f_2=g_1$, completing all the roots in $\F_q$.
	\qed
\end{proof}

\subsection{The case $c=3$.} 

Le us move to a larger value of $c$ and deal with the case $c=3$. 

This representation of the Bezout identity in Definition~\ref{def:bezoutiden} is not unique. For  univariate polynomials we have the following lemma. We include proofs for clarity in the exposition.
\begin{lema}
	\label{lem:mindegbezoutone}
 Let $u, v\in \F_q[x]$ be two non constant polynomials relatively prime. Then there exist polynomials $a, b\in\F_q[x]$
	such that
	$$
	1 = au - bv,
	$$	
	with $\deg(a)< \deg(v)$ and $\deg(b) < \deg(u)$.
\end{lema}

\begin{proof}
	By the Bezout identity we have that, for some $\hat a,\hat b\in F_q[x]$,
	\begin{equation}
	\hat au - \hat bv=1
	\end{equation}
	If $\deg(\hat a)< \deg(v)$ then the theorem follows because 
	\begin{equation}
	\label{eq:bezdegeq}
		\deg(\hat a)+ \deg(u) = \deg(\hat b)+ \deg(v).
	\end{equation}

	Assume $\deg(\hat a)\geq \deg(v)$. It is clear that the pair $a = \hat a + tv$, $b = \hat b + tu$ also satisfy  Bezout's identity for any $t\in\F_q[x]$. Dividing $\hat a$ by $v$,we get  $\hat a= q_av + r_a$ with $\deg(r_a)<\deg(v)$ and  taking $t=-q_a$, we have that $a=r_a$, $b = \hat b - q_au$ and the result follows using the same reasoning as in~\eqref{eq:bezdegeq}, since $a$ and $b$ also satisfy Bezout's identity.
	\qed
	\end{proof}

\begin{coro}
\label{cor:mindegbezout}
Let $u, v\in\F_q[x]$ be two non constant polynomials relatively prime. Let $z$ a polynomial such that 
$\deg(z) < \deg(u) + \deg(v)$. Then, we can express $z$ as
$$
z = au - bv,
$$	
for $a,b\in\F_q[x] $with $\deg(a)< \deg(v)$ and $\deg(b) < \deg(u)$.
\end{coro}

\begin{proof}
By the previous lemma we have that there exist $\hat a$ and $\hat b$, such that
\begin{equation}
\hat au - \hat bv=1
\end{equation}
with $\deg(\hat a)< \deg(v)$ and $\deg(\hat b) < \deg(u)$.
Therefore,

\begin{equation}
\label{eq:bezouttimes}
\hat azu - \hat bzv=z
\end{equation}

It is clear that
\begin{align}\label{eq:bezoutotherpol}
a&=\hat az+tv \nonumber\\
b&=\hat bz+tu,
\end{align}
satisfy
\begin{equation}
\label{eq:bezouttimesab}
au - bv=z,
\end{equation}
for any $t\in\F_q[x]$.

Now, since we can express $\hat az = vq_a + r_a$, we substitute $t=-q_a$ in~\eqref{eq:bezoutotherpol} obtaining,
\begin{align}\label{eq:bezoutresiduepol}
a&=r_a \nonumber\\
b&=\hat b z-q_au.
\end{align}
The polynomials $a$ and $b$ in~\eqref{eq:bezoutresiduepol} satisfy~\eqref{eq:bezouttimesab}. We have the following cases.
If $\deg(r_au)\leq \deg(z)$ then $\deg(bv)\leq \deg(z)$ and so
$\deg(b)\leq \deg(z)-\deg(v)<\deg(u)$. On the other hand, if
$\deg(r_au)> \deg(z)$ then we have $\deg(r_au)= \deg(bv)$ and so
$0<\deg(v)-\deg(a)= \deg(u)-\deg(b)$.
\qed
\end{proof}

We express our result in the form of a theorem.

\begin{teo}
\label{thm:case3} Let  $q$  be a power of a prime,  $c=3$ and  $d<q-
	2\left[\frac {q}{8}\right]$. A Reed Solomon code of length $q$ over $\F_q$ with distance $d$ is not $3$-separating.
\end{teo}

\begin{proof}
Let $q=8l+r$ where $0\le r<8$. We make a partition of $\F_q$ in nine disjoint sets $U_{i,j}\subset \F_q$,  $1\le i, j\le 3$, as follows: $U_{1,1}=\emptyset$,  $r$ sets  of size $l+1$ and the other remaining $8-r$ sets  of size $l$. Observe that $\sum_{i,j}|U_{i,j}|=q.$

Now let $u_{1,1}=-1$, $u_{i,j}=\prod_{\alpha\in U_{i,j}}(x-\alpha)$ for $(i,j)\ne (1,1)$ and take $v_{i,j}$ for $i\ne 1\ne j$  the solutions of smaller degree of the Bezout equations
\begin{eqnarray*}\label{eq:bezout}
	&&v_{3,3}u_{3,3}-v_{3,2}u_{3,2}=1,\\
	&&v_{2,3}u_{2,3}-v_{2,2}u_{2,2}=1.
\end{eqnarray*}
Note that according to  Corollary~\ref{cor:mindegbezout},
the degree of $v_{2,2}$, $v_{2,3}$, $v_{3,2}$ and $v_{3,3}$ is less than $l+1$.

Then we proceed similarly to define $v_{2,1}, v_{3,1}, v_{1,2}$ and $v_{1,3}$ to be  the solutions of smaller degree of the Bezout equations
\begin{eqnarray*}\label{eq:bezout2}
	&&v_{3,3}u_{3,3}-v_{2,3}u_{2,3}=v_{3,1}u_{3,1}-v_{2,1}u_{2,1}\\
		&&v_{3,3}u_{3,3}-v_{3,2}u_{3,2}=v_{1,3}u_{1,3}-v_{1,2}u_{1,2}.
\end{eqnarray*}
Again, by Corollary~\ref{cor:mindegbezout} the degree of  $v_{2,1}, v_{3,1}, v_{1,2}$ and $v_{1,3}$ is less than $l+1$.
Finally, take $v_{1,1}=v_{3,3}u_{3,3}-v_{3,1}u_{3,1}-v_{1,3}u_{1,3}.$ Now we  define 
\begin{equation}
\begin{array}{lc} g_1=-v_{3,1}u_{3,1}, & \quad f_1=g_1+v_{1,1}u_{1,1},\\
g_2=-v_{3,2}u_{3,2}, &\quad f_2=g_2+v_{2,2}u_{2,2},\\
g_3=-v_{3,3}u_{3,3}, & \quad f_3=g_3+v_{3,3}u_{3,3}.\end{array}
\end{equation}
By definition, we have
$$
(x^q-x)\left | \prod_{1\le i,j\le 3}v_{i,j}u_{i,j}=\prod_{1\le i,j\le 3} (f_i-g_j).\right.
$$
By construction we see that deg$(f_i)\le 2l+1$, deg$(g_i)\le 2l+1$ and hence the result follows.
\qed
\end{proof}

\section{The general case.}

In order to obtain stronger results we need to deal with larger values of both $c$ and the minimum distance. The following theorem generalizes Theorem~\ref{teo:2gen} for $c\geq 2$.

\begin{teo} Let  $q$  be a power of a prime,  $c\ge 2$ and $d< q-
\left[\frac {q}{2c-1}\right]$. A Reed Solomon code  over $\F_q$  with distance $d$  is not $c$-SEP.
\end{teo}
\begin{proof} Let $q=(2c-1)l+r$ where $0\le r< 2c-1$,  and consider for any $1\le i \le 2c-1$, $A_{i}\subset \F_q$ disjoint sets such that: $r$ of the sets are of  size $|A_{i}|=l+1$ and  the remaining $2c-r-1$ are of size $|A_{i}|=l$. Now let $p_{i}=\prod_{\alpha\in A_{i}}(x-\alpha)$. Observe that $\cup_{i}A_{i}=\F_q$.
The following polynomials, of degree at most $l+1$ evaluate to code words of a code that is not $c$-SEP.
\begin{eqnarray*}
f_1&=&0,\\
f_{i+1}&=&p_{c+i}-p_{i},\quad \text{ for } 1\le i\le c-1\\
g_i&=&-p_{i},\quad \text{ for } 1\le i\le c.
\end{eqnarray*}
Indeed, observe that $f_1(\alpha)=g_i(\alpha)$ for any $\alpha\in A_i, 1\le i\le c$,  while $f_i(\alpha)=g_i(\alpha)$ for any $\alpha\in A_{c+i}$ for any $2\le i\le c$.
\qed
\end{proof}

To cope with a larger minimum distance, we would like to extend Theorem~\ref{thm:case3}. Unfortunately, the generalization is not immediate  because when $c$ grows, the degree of the polynomials $v_{i,j}$ blows up. To obtain stronger results we need to take advantage of the structure of the field over which the code is defined. In this case we are able to state a result for a minimum distance matching the conjectured one, but limited to certain parameters of the code. 

\begin{teo}\label{teo:cisquare} Let $m^2|q-1$. Then, for any $c\ge m$, there exist a non extended Reed Solomon code  over $F_{q}$ with distance $d=q-\frac{q-1}{m^2}$ that is not $c$-SEP.
\end{teo}

{\bf Proof.} Let $\alpha$ be a primitive root of the multiplicative group $\F_q^*$. Consider $g_i=\alpha^{i\frac{(q-1)}{d^2}}$,  and 
$f_i=\alpha^{\frac{-i(q-1)}{d}}x^{\frac{q-1}{d^2}}$, $i=0,\dots,c-1$. Now, every element of $\F_q^*$ can be written as  $\alpha_{r,s}=\alpha^{ld^2+rd+s}$ for some $0\le s, r<d$, and certain integer $l$. Then we have clearly $f_r(\alpha_{r,s})=g_s$, proving the result.
\

\begin{coro}\label{cor:1c2} Let  $c$  be any integer and $q\equiv 1 \pmod{c^2}$.  There exist  a non extended Reed Solomon code with distance $d=q-\left[\frac{q}{c^2}\right]$ over $F_{q}$  which is $(c,c)-$ inseparable.
\end{coro}
{\bf Proof.} Simply note that if $c^2|q-1$, then  $\frac{q-1}{c^2}=\left[\frac{q}{c^2}\right]$

\

\begin{coro} For any $p$ and $c$ there exist infinitely many $q=p^e$ such that $\F_q$ admits a non extended Reed-Solomon
	code of distance $d=q-\left[\frac{q}{c^2}\right]$.
\end{coro}
{\bf Proof.} Simply note that by Fermat's little theorem $p^{\varphi(c^2)}\equiv 1\pmod {c^2}$, so the result follows for any $e=k\varphi(c^2)$, $k\in\N$, applying the previous theorem.

\section{The ``linear'' case $d=q-1$}

The case presented in this section is already dealt with in Corollary~\ref{cor:sepcgsqrtk1}. We include it here, because the proofs
provide new ways to approach a complete solution to the problem.

The first result we prove is a straight forward application of the following 
theorem of J. Cilleruelo.

\begin{teo} 
	\label{thm:cille}
	(Cilleruelo) Let $\alpha$ be a generator of $\F_q^*$. Then 
	$$
	\{\alpha^i-\alpha^j\, :\, 0\le i,j\le 2q^{3/4}\}=\F_q.
	$$
\end{teo}

Now, we have

\begin{teo}\label{teo:34} Let $c\ge 2q^{3/4}$. Then, $[n,k,d]$ Reed Solomon codes over $F_{q}$ of length $n$ and distance $d=q-1$ are not $c$-SEP.
\end{teo}

{\bf Proof.} 
Let us note first that, since $c\ge 2q^{3/4}$, then $c^2>q$ and then $q-[q/c^2]-1=q-1$. So we need a code with distance $d=q-1$. This means that we need to find two families of polynomials of size $c$ each, with  all  polynomials of degree at most $1$, such that 
$$
\{f_1(\alpha),\dots,f_c(\alpha)\}\cap \{g_1(\alpha),\dots,g_c(\alpha)\}\ne\emptyset
$$ 
for any $\alpha\in\F_q$. This is the same as 
$$
(x^q-x)\left | \prod_{1\le i,j\le c} (f_i-g_j).\right.
$$
 Now, let $\alpha$ be a generator of $\F_q^*$, and consider $f_i=x-\alpha^i$, $g_i=-\alpha^i$, for $i=1,\dots, c$. 
It follows that,
$$
\prod_{1\le i,j\le c}(f_i-g_j)=\prod_{1\le i,j\le c}(x-(\alpha^i-\alpha^j))
$$
and by the previous theorem we trivially have
$$
(x^q-x)\left|\prod_{1\le i,j\le c}(f_i-g_j).\right.
$$

But we can make it better.

\begin{teo}\label{teo:factor} Suppose that $q-1=rs$ such that $(r,s)=1$ and suppose $c>\max\{r,s\}$. Then, $[n,k,d]$ Reed Solomon codes with distance $d=q-1$ over $F_{q}$  which is not $c$-SEP.
\end{teo}
{\bf Proof.} Let $q-1=rs$ such that $(r,s)=1$, $\alpha$ a generator of $\F_q$ and consider the sets $A=\{1,\alpha^r,\dots, \alpha^{r(s-1)}\}$ and $B=\{1,\alpha^s,\dots, \alpha^{s(r-1)}\}$. Then, all the quotients $a/b$ with $a\in A$ and $b\in B$ are distinct. Indeed, suppose $\alpha^{ri}/\alpha^{sj}=\alpha^{rI}/\alpha^{sJ}$. Then $\alpha^{r(i-I)}/\alpha^{s(J-j)}$, and so $\alpha^{r(i-I)-s(J-j)}=1$ but, since $\alpha$ is a generator,  this is only possible either  if $r(i-I)-s(J-j)=0$ or else if 
 $(q-1)|r(i-I)-s(J-j)$. In any of the two cases, since $r|q-1$, we have $r|s(J-j)$ and since $(r,s)=1$, then $r|(J-j)$ but this is impossible, since $|J-j|<r$, unless $J=j$, and then $i=I$.

 Now, consider polynomials $f_i=\alpha^{ri}x$, $g_j=\alpha^{sj}$ with $0\leq i \leq s-1$ and $0\leq j \leq r-1$. We can do that since $l<c$. By the previous argument, the roots of $f_i-g_j$ are all distinct and we have $rs=q-1$ distinct roots. Since $r<c$ we can just add the root missing by adding a polynomial $g_{r}$.
 
 Observe that, since $c>s$, then $c^2>rs=q-1$, and we can suppose $c^2>q$ since the case $c^2=q$ is already proved. So $[q/c^2]=0$ and the correct distance is $d=q-1$, so we have to consider linear polynomials. (By the theorem of catalan $q^x-p^y=1$ only if $3^2-2^3$) 

\

The previous theorem improves Theorem \ref{teo:34} when $q$ is an even power. Indeed, in the case in which $q=p^{2t}$, then either $(p^t-1)/2$ is odd or $(p^t+1)/2$ is odd. Without loss of generality, assume $(p^t-1)/2$ is odd. Then,   we can take $r=(p^t-1)/2$ and $s=2(p^t+1)$ and so $s\le 4r+4$. Therefore $q-1=rs\le 4r^2+4r$ or $q\le 4r^2+4r+1=(2r+1)^2$ which gives $r\ge(\sqrt q-1)/2$. Then, $q-1=rs\ge  ((\sqrt q-1)/2)s$ which gives $s\le \frac{q-1}{(\sqrt q-1)/2}=2(\sqrt q+1)$. Hence, since Theorem~\ref{teo:factor} assumes $c>s$ then for any $c>2(\sqrt q+1)$ we have that Reed Solomon code with distance $d=q-1$ over $F_{q}$  are not $c$-SEP, improving  Theorem \ref{teo:34}. 

\

In general, the theorem provides a general bound on $c$, depending on the factorization on the exponent. However, in the case of a sophie germain prime, $q-1=2p$ where q and $p$ are primes, then Theorem \ref{teo:factor} only gives $c\ge q/2$.

\section{Conclusion}

The aim of the paper is to find out whether or not there exist values of the minimum distance for which a Reed-Solomon is $c$-SEP but not $c$-TA. We start the presentation by considering a sufficiently small value to the minimum distance. For this much convenient
value, we prove that codes do not posses the frameproof property, let alone the separating one. For the cases $c=2$ and $c=3$, we improve this almost naive result by introducing   to our discourse both polynomial interpolation and   Bezout's identity.  

The approach for case $c=3$ does not generalize to larger values of $c$. In order to deal with the general case, we resort
to the structure of $\F_q$, the finite field over which the code is defined. This allows us to prove an assertion for all $c$,  whenever $q\equiv 1 \pmod{c^2}$.  Along the
same line of reasoning, we provide an alternative proof of existing results by applying an elegant theorem about the generator of the multiplicative group of $\F_q$.

Our presentation  shows that  for the general case, a constructive proof  is by no means trivial. This is because, when using the structure of the field defining the code one can not encircle all cases and  cases without  ``structure'' do not seem to follow any common pattern. So, although the problem is algebraic in nature, it seems that an existence proof should be considered. 

\bibliographystyle{plain}
\bibliography{sepcodes}

\end{document}